\documentclass[journal]{IEEEtran}
\ifCLASSINFOpdf
\else
\fi
\usepackage[cmex10]{amsmath}
\usepackage{amsfonts}
\usepackage{amssymb}

\newtheorem{theorem}{\bf Theorem}

\newtheorem{corollary}[theorem]{\bf Corollary}

\newtheorem{definition}[theorem]{\bf Definition}

\newtheorem{proposition}[theorem]{\bf Proposition}

\newenvironment{proof}[1][Proof]{\noindent\textbf{#1.} }{\ \rule{0.5em}{0.5em}}

\begin{document}
%
\title{$\alpha$-divergence derived as\\ the generalized rate function in a power-law system}

%
%
%

\author{Hiroki~Suyari and Antonio Maria~Scarfone%
\thanks{H. Suyari is with Graduate School of Advanced Integration Science,
Chiba University, 1-33 Yayoi, Inage, Chiba, Chiba 263-8522, Japan
 e-mail: suyari@faculty.chiba-u.jp}
\thanks{A.M. Scarfone is with Istituto dei Sistemi Complessi (ISC-CNR) 
c/o Dipartimento di Scienza Applicatae Tecnologia,
Politecnico di Torino, Corso Duca degli Abruzzi 24,
10129 Torino, Italy
 e-mail: antonio.scarfone@polito.it}}


%


\maketitle

\begin{abstract}
The generalized binomial distribution in Tsallis statistics (power-law system)
is explicitly formulated from the precise $q$-Stirling's formula. The $\alpha
$-divergence (or $q$-divergence) is uniquely derived from the generalized
binomial distribution in the sense that when $\alpha\rightarrow-1$ (i.e.,
$q\rightarrow1$) it recovers KL divergence obtained from the standard binomial
distribution. Based on these combinatorial considerations, it is shown that
$\alpha$-divergence (or $q$-divergence) is appeared as the generalized rate
function in the large deviation estimate in Tsallis statistics.
\end{abstract}


%
\IEEEpeerreviewmaketitle

\section{Introduction}

The large deviation principle (LDP for short) has mathematically presented and
quantified the asymptotic behavior of the probabilities of rare events in many
stochastic phenomena. It has brought about deep significant insights for
understanding of each phenomena \cite{Varadhan66}\cite{Varadhan84}\cite{DZ98}.
The LDP covers quite broad areas ranging from the fundamentals in probability
theory and statistics to its applications such as statistical physics
\cite{Ellis85}\cite{Touchette09}, risk management \cite{AK96}, information
theory \cite{CoverThomas06} and so on. In most of theoretical results in LDP,
the assumption of \textquotedblleft i.i.d. (independent and identically
distributed)\textquotedblright\ for random variables is used. This assumption
leads to the discussion on the exponential decay of rare events in stochastic
phenomena with great help of many well-established theoretical results based
on \textquotedblleft i.i.d.\textquotedblright\ assumption. This strong
\textquotedblleft i.i.d.\textquotedblright\ assumption\ has been often tried
to be weakened in many studies. One of the reasons is that actual observations
generally do not satisfy i.i.d. assumptions. A typical and well-known example
is power-law behavior often observed in strongly correlated systems. In these
cases we take Tsallis statistics as one of such power-law systems because its
mathematical foundations has been widely explored \cite{Suyari06}.

Along similar studies on LDP related with Tsallis statistics, there are a few
papers such as \cite{RuizTsallis2012} and \cite{NaudtsSuyari2014}.
The paper \cite{RuizTsallis2012} discusses the possibility of LDP for the strongly correlated
random variables in Tsallis statistics.
They consider the correlated coin tossing model
based on the $q$-Gaussian distribution and numerically evaluate the
possibility of a $q$-generalization of LDP for a given $q$-divergence. On the
other hand, our present paper does not require the $q$-Gaussian distribution
and the $q$-divergence in advance for the large deviation estimate. Our
approach is completely analytical starting from the fundamental nonlinear
differential equation $dy/dx=y^{q}$ only, and the $q$-divergence is naturally
derived from the formulation of the generalized binomial distribution
($q$-binomial distribution), which results in an appearance of the
$q$-divergence as the rate function in our main result (i.e., Theorem 14).
Thus our approach and its results are quite different from
\cite{RuizTsallis2012}.
The paper
\cite{NaudtsSuyari2014} discusses the application of the $q$-exponential
functions to the standard LDP based on the i.i.d. assumptions, which obviously
differs from our present work.

In this paper, we apply our combinatorial formulations in Tsallis statistics
to the large deviation estimate in it. This paper consists of the five
sections including this introduction. In the section two we briefly review the
derivation of Tsallis entropy as the unique entropy corresponding to
$dy/dx=y^{q}$. In the course of the derivation of Tsallis entropy, the rough
$q$-Stirling's formula is derived and applied. In the formulation of the
$q$-Stirling's formula, the precise $q$-Stirling's formula has been already
given in \cite{Suyari06}, which is useful for the formulation of the
$q$-binomial distribution in Tsallis statistics. The derivation of the
$q$-binomial distribution is given in the section three. The $q$-binomial
distribution has the nice correspondence with the $q$-divergence (or $\alpha
$-divergence), which is given in the following section four. The one-to-one
correspondence between the $q$-binomial distribution and the $q$-divergence
(or $\alpha$-divergence) is applied to the large deviation estimate in a
combinatorial approach. The final section is devoted to the conclusion.

\section{Tsallis entropy and $q$-Stirling's formula uniquely determined from
$dy/dx=y^{q}$}

Our approach starts from the fundamental nonlinear differential equation
$dy/dx=y^{q}$ only, instead of Tsallis entropy. This approach provides us with
most of the theoretical results in Tsallis statistics shown in \cite{Suyari06}%
\cite{SuyariWada08}. (See \cite{Ts88}\cite{Tsallis09}\cite{Naudts11} for the
review of the studies in Tsallis statistics and the maximum entropy principle approach.)

The solution to $dy/dx=y^{q}$ yields the $q$-exponential $\exp_{q}\left(
x\right)  $ as the inverse function of the $q$-logarithm $\ln_{q}x$, which are
respectively defined as follows:

\begin{definition}
($q$-logarithm, $q$-exponential) The $q$-logarithm $\ln_{q}x:\mathbb{R}%
^{+}\rightarrow\mathbb{R}$ and the $q$-exponential $\exp_{q}\left(  x\right)
:\mathbb{R}\rightarrow\mathbb{R}$ for $x\in\mathbb{R}$ satisfying $1+\left(
1-q\right)  x>0$ are respectively defined by%
\begin{align}
\ln_{q}x &  :=\frac{x^{1-q}-1}{1-q},\label{q-logarithm}\\
\exp_{q}\left(  x\right)   &  :=\left[  1+\left(  1-q\right)  x\right]
^{\frac{1}{1-q}}.\label{q-exponential}%
\end{align}

\end{definition}

Then a new product $\otimes_{q}$ to satisfy the following identities as the
$q$-exponential law is introduced.%
\begin{align}
\ln_{q}\left(  x\otimes_{q}y\right)   &  =\ln_{q}x+\ln_{q}%
y,\label{requirement1}\\
\exp_{q}\left(  x\right)  \otimes_{q}\exp_{q}\left(  y\right)   &  =\exp
_{q}\left(  x+y\right)  . \label{requirement2}%
\end{align}
For this purpose, the new multiplication operation $\otimes_{q}$ is introduced
in \cite{NMW03}\cite{Bo03} (See also \cite{Scarfone13}). The concrete forms of
the $q$-logarithm and $q$-exponential are given in (\ref{q-logarithm}) and
(\ref{q-exponential}), so that the above requirement (\ref{requirement1}) or
(\ref{requirement2}) as the $q$-exponential law leads to the definition of
$\otimes_{q}$ between two positive numbers.

\begin{definition}
($q$-product) For $x,y\in\mathbb{R}^{+}$ satisfying $\,x^{1-q}+y^{1-q}-1>0$,
the $q$-product $\otimes_{q}$ is defined by%
\begin{equation}
x\otimes_{q}y:=\left[  x^{1-q}+y^{1-q}-1\right]  ^{\frac{1}{1-q}}.
\label{def of q-product}%
\end{equation}

\end{definition}

The $q$-product recovers the usual product such that $\underset{q\rightarrow
1}{\lim}\left(  x\otimes_{q}y\right)  =xy$.

By means of the $q$-product (\ref{def of q-product}), the $q$-factorial is
naturally defined in the following form \cite{Suyari06}.

\begin{definition}
($q$-factorial) For a natural number $n\in\mathbb{N}$ and $q\in\mathbb{R}^{+}%
$, the $q$-factorial $n!_{q}$ is defined by%
\begin{equation}
n!_{q}:=1\otimes_{q}\cdots\otimes_{q}n.\label{def of q-kaijyo}%
\end{equation}

\end{definition}

Thus, we concretely compute $q$-Stirling's formula.

\begin{theorem}
(rough $q$-Stirling's formula) Let $n!_{q}$ be the $q$-factorial defined by
(\ref{def of q-kaijyo}). The rough $q$-Stirling's formula $\ln_{q}\left(
n!_{q}\right)  $ is computed as follows:%
\begin{equation}
\ln_{q}\left(  n!_{q}\right)  =\left\{
\begin{array}
[c]{ll}%
\dfrac{n\ln_{q}n-n}{2-q}+O\left(  \ln_{q}n\right)  \quad & \text{if}\quad
q\neq2,\\
n-\ln n+O\left(  1\right)  & \text{if}\quad q=2.
\end{array}
\right.  \label{rough q-Stirling}%
\end{equation}

\end{theorem}

See \cite{Suyari06} for the proof.

Similarly to the $q$-product, $q$-ratio is introduced from the requirements:
\begin{align}
\ln_{q}\left(  x\oslash_{q}y\right)   &  =\ln_{q}x-\ln_{q}y,\\
\exp_{q}\left(  x\right)  \oslash_{q}\exp_{q}\left(  y\right)   &  =\exp
_{q}\left(  x-y\right)  .
\end{align}
Then we define the $q$-ratio as follows.

\begin{definition}
($q$-ratio) For $x,y\in\mathbb{R}^{+}$ satisfying $x^{1-q}-y^{1-q}+1>0$, the
inverse operation to the $q$-product is defined by
\begin{equation}
x\oslash_{q}y:=\left[  x^{1-q}-y^{1-q}+1\right]  ^{\frac{1}{1-q}}%
\end{equation}
which is called $q$-ratio in \cite{Bo03}.
\end{definition}

The $q$-product, $q$-factorial and $q$-ratio are applied to the definition of
the $q$-multinomial coefficient \cite{Suyari06}.

\begin{definition}
($q$-multinomial coefficient) For $n=\sum_{i=1}^{k}n_{i}$ and $n_{i}%
\in\mathbb{N}\,\left(  i=1,\cdots,k\right)  ,$ the $q$-multinomial coefficient
is defined by%
\begin{equation}
\left[
\begin{array}
[c]{ccc}
& n & \\
n_{1} & \cdots & n_{k}%
\end{array}
\right]  _{q}:=\left(  n!_{q}\right)  \oslash_{q}\left[  \left(  n_{1}%
!_{q}\right)  \otimes_{q}\cdots\otimes_{q}\left(  n_{k}!_{q}\right)  \right]
.\label{def of q-multinomial coefficient}%
\end{equation}

\end{definition}

From the definition (\ref{def of q-multinomial coefficient}), it is clear that%
\begin{equation}
\underset{q\rightarrow1}{\lim}\left[
\begin{array}
[c]{ccc}
& n & \\
n_{1} & \cdots & n_{k}%
\end{array}
\right]  _{q}=\left[
\begin{array}
[c]{ccc}
& n & \\
n_{1} & \cdots & n_{k}%
\end{array}
\right]  =\frac{n!}{n_{1}!\cdots n_{k}!}.
\end{equation}
Throughout this paper, we consider the $q$-logarithm of the $q$-multinomial
coefficient to be given by
\begin{equation}
\ln_{q}\left[
\begin{array}
[c]{ccc}
& n & \\
n_{1} & \cdots & n_{k}%
\end{array}
\right]  _{q}=\ln_{q}\left(  n!_{q}\right)  -\ln_{q}\left(  n_{1}!_{q}\right)
\cdots-\ln_{q}\left(  n_{k}!_{q}\right)  . \label{lnq-q-multinomial}%
\end{equation}

Based on these fundamental formulas, we obtain the one-to-one correspondence
(\ref{2-q-correspondence}) between the $q$-multinomial coefficient and Tsallis
entropy as follows \cite{Suyari06}.

\begin{theorem}
When $n\in\mathbb{N}$ is sufficiently large, the $q$-logarithm of the
$q$-multinomial coefficient coincides with Tsallis entropy in the following
correspondence:%
\begin{align}
&  \ln_{q}\left[
\begin{array}
[c]{ccc}
& n & \\
n_{1} & \cdots & n_{k}%
\end{array}
\right]  _{q}\nonumber\\
&  \simeq\left\{
\begin{array}
[c]{ll}%
\dfrac{n^{2-q}}{2-q}\cdot S_{2-q}\left(  \dfrac{n_{1}}{n},\cdots,\dfrac{n_{k}%
}{n}\right)   & \text{if}\quad q>0,\,\,q\neq2\\
-S_{1}\left(  n\right)  +\sum\limits_{i=1}^{k}S_{1}\left(  n_{i}\right)   &
\text{if}\quad q=2
\end{array}
\right.  \label{2-q-correspondence}%
\end{align}
where $S_{q}$ is Tsallis entropy defined by%
\begin{equation}
S_{q}\left(  {p_{1},\ldots,p_{k}}\right)  :={\frac{{{1-\sum\limits_{i=1}%
^{k}{p_{i}^{q}}}}}{{q-1}}}\label{Tsallis entropy}%
\end{equation}
and $S_{1}\left(  n\right)  :=\ln n.$
\end{theorem}

See \cite{Suyari06} for the proof. In this way, Tsallis entropy is determined
as the unique entropy corresponding to the fundamental nonlinear differential
equation $dy/dx=y^{q}$.

Clearly, the additive duality \textquotedblleft$q\leftrightarrow2-q$%
\textquotedblright\ appears in the above one-to-one correspondence
(\ref{2-q-correspondence}). Other dualities such as the multiplicative duality
\textquotedblleft$q\leftrightarrow1/q$\textquotedblright, $q$-triplet and
multifractal triplet appears as special cases of the more generalized
correspondence \cite{SuyariWada08}. Apart from these derivations, a typical
and fundamental application of the $q$-product is the derivation of the
$q$-Gaussian distribution through the maximum likelihood principle
\cite{Suyari04-LawofError}.

\section{The generalized binomial distribution derived from the precise
$q$-Stirling's formula}

In the previous section, the rough $q$-Stirling's formula
(\ref{rough q-Stirling}) is applied to the derivation of Tsallis entropy as
the unique entropy corresponding to $dy/dx=y^{q}$. In order to define the
generalized binomial distribution, the precise $q$-Stirling's formula is required.

\begin{theorem}
(precise $q$-Stirling's formula) Let $n!_{q}$ be the $q$-factorial defined by
(\ref{def of q-kaijyo}). The precise $q$-Stirling's formula $\ln_{q}n!_{q}$ is
computed by:%
\begin{equation}
\ln_{q}n!_{q}\simeq\left\{
\begin{array}
[c]{ll}%
n-\frac{1}{2n}-\ln n-\frac{1}{2}-\delta_{2} & \left(  q=2\right) \\
\left(  \frac{n}{2-q}+\frac{1}{2}\right)  \ln_{q}n-\frac{n}{2-q}+c_{q} &
\left(  q\neq2\right)
\end{array}
\right.  \label{precise q-Stirling_cq}%
\end{equation}
where $c_{q}:=\frac{1}{2-q}-\delta_{q}$, $\delta_{q}$ is a function of $q$
only, and $\delta_{1}=1-\ln\sqrt{2\pi}$.
\end{theorem}

The proof is given in \cite{Suyari06}. Note that the terms $c_{q}$ and
$\frac{1}{2}$ do not depend on $n$, so that the precise $q$-Stirling's
formula (\ref{precise q-Stirling_cq}) recovers the rough $q$-Stirling's
formula (\ref{rough q-Stirling}) if $c_{q}$ and $\frac{1}{2}$ are ignored.

\begin{proposition}%
\begin{align}
\ln_{q}\left[
\begin{array}
[c]{l}%
n\\
k
\end{array}
\right]  _{q}  &  \simeq-c_{q}+\frac{1}{2}\left(  \ln_{q}n-\ln_{q}k-\ln
_{q}\left(  n-k\right)  \right) \nonumber\\
&  +\frac{n^{2-q}}{2-q}S_{2-q}\left(  \frac{k}{n},1-\frac{k}{n}\right)
\label{q-log_q-bino-coef}%
\end{align}

\end{proposition}

This is easily proved by the straightforward computation using the precise
$q$-Stirling's formula (\ref{precise q-Stirling_cq}).

For easy understanding of the derivation of the generalized binomial
distribution, consider the case $q=1$ in (\ref{q-log_q-bino-coef}).%
\begin{align}
&  \ln\left[
\begin{array}
[c]{l}%
n\\
k
\end{array}
\right] \nonumber\\
&  \simeq-\ln\sqrt{2\pi}+\frac{1}{2}\ln\frac{n}{k\left(  n-k\right)  }%
+nS_{1}\left(  \frac{k}{n},1-\frac{k}{n}\right)
\label{stand-log_stand-bino-coef-1}\\
&  =\ln\frac{1}{\sqrt{2\pi}}\sqrt{\frac{n}{k\left(  n-k\right)  }}+\left(
-k\ln\frac{k}{n}-\left(  n-k\right)  \ln\left(  1-\frac{k}{n}\right)  \right)
\label{stand-log_stand-bino-coef}%
\end{align}
That is, we have%
\begin{equation}
\left[
\begin{array}
[c]{l}%
n\\
k
\end{array}
\right]  \left(  \frac{k}{n}\right)  ^{k}\left(  1-\frac{k}{n}\right)
^{n-k}\simeq\frac{1}{\sqrt{2\pi}}\sqrt{\frac{n}{k\left(  n-k\right)  }}.
\label{special form of the standard}%
\end{equation}
The left side of (\ref{special form of the standard}) is the special form of
the standard binomial distribution $\left[
\begin{array}
[c]{l}%
n\\
k
\end{array}
\right]  r^{k}\left(  1-r\right)  ^{n-k}$ in the sense of $r=\frac{k}{n}$. For
generalization of the standard binomial distribution, the most important
observation in this computation is that the term:%
\begin{equation}
nS_{1}\left(  \frac{k}{n},1-\frac{k}{n}\right)  =-k\ln\frac{k}{n}-\left(
n-k\right)  \ln\left(  1-\frac{k}{n}\right)  \label{last term q=1}%
\end{equation}
in (\ref{stand-log_stand-bino-coef-1}) and (\ref{stand-log_stand-bino-coef})
corresponds to $r^{k}\left(  1-r\right)  ^{n-k}$ in the standard binomial
distribution $\left[
\begin{array}
[c]{l}%
n\\
k
\end{array}
\right]  r^{k}\left(  1-r\right)  ^{n-k}$ by replacement $r=\frac{k}{n}$. More
precisely, (\ref{last term q=1}) coincides with%
\begin{equation}
-\ln r^{k}\left(  1-r\right)  ^{n-k}=-k\ln r-\left(  n-k\right)  \ln\left(
1-r\right)  \label{replacement r=k/n}%
\end{equation}
by the replacement. (Compare the right sides in (\ref{last term q=1}) and
(\ref{replacement r=k/n}).)

This correspondence is also applied to (\ref{q-log_q-bino-coef}). The last
term on the right hand of (\ref{q-log_q-bino-coef}) is computed as%
\begin{align}
&  \frac{n^{2-q}}{2-q}S_{2-q}\left(  \frac{k}{n},1-\frac{k}{n}\right)
\nonumber\\
&  =\frac{1}{2-q}\left(  -k^{2-q}\ln_{2-q}\frac{k}{n}-\left(  n-k\right)
^{2-q}\ln_{2-q}\left(  1-\frac{k}{n}\right)  \right)  .
\end{align}
Applying the replacement $r=\frac{k}{n}$ in this formula to
(\ref{q-log_q-bino-coef}) as similarly as the replacement $r=\frac{k}{n}$ in
the standard case (\ref{stand-log_stand-bino-coef}), the generalized binomial
distribution $b_{q}\left(  k;n,r\right)  $ is defined.

\begin{definition}
For given $n,k\left(  \leq n\right)  \in\mathbb{N}$ and $r\in\left(
0,1\right)  $, if the $q$-logarithm of the probability mass function
$b_{q}\left(  k;n,r\right)  $ is given by%
\begin{align}
&  \ln_{q}b_{q}\left(  k;n,r\right)  =\ln_{q}\left[
\begin{array}
[c]{l}%
n\\
k
\end{array}
\right]  _{q}\nonumber\\
&  \quad+\frac{1}{2-q}\left(  k^{2-q}\ln_{2-q}r+\left(  n-k\right)  ^{2-q}%
\ln_{2-q}\left(  1-r\right)  \right)  +C_{q}%
\label{def_lnq_q-binomial distribution}%
\end{align}
with%
\begin{equation}
\sum_{k=0}^{n}b_{q}\left(  k;n,r\right)  =1\text{,\quad}1+\left(  1-q\right)
C_{q}>0\text{\quad and\quad}C_{1}=0,
\end{equation}
$b_{q}\left(  k;n,r\right)  $ is the $q$-binomial distribution.
\end{definition}

The $q$-multinomial distribution $m_{q}\left(  n_{1},\cdots,n_{k}%
;n,r_{1},\cdots,r_{k}\right)  $ is easily defined in a similar way.%
\begin{align}
&  \ln_{q}m_{q}\left(  n_{1},\cdots,n_{k};n,r_{1},\cdots,r_{k}\right)
\nonumber\\
&  =\ln_{q}\left[
\begin{array}
[c]{ccc}
& n & \\
n_{1} & \cdots & n_{k}%
\end{array}
\right]  _{q}+\frac{1}{2-q}\sum_{i=1}^{k}n_{i}^{2-q}\ln_{2-q}r_{i}%
+C_{q}\label{def_lnq_q-multinomial distribution}%
\end{align}
where $n=\sum_{i=1}^{k}n_{i}$, $n_{i}\in\mathbb{N}\,\left(  i=1,\cdots
,k\right)  $, and $1=\sum_{i=1}^{k}r_{i}$.

The explicit form of the $q$-binomial distribution $b_{q}\left(  k;n,r\right)
$ can be written, but the form of $\ln_{q}b_{q}\left(  k;n,r\right)  $ in
(\ref{def_lnq_q-binomial distribution}) is much simpler and more useful for
applications. Note that the $q$-binomial distribution $b_{q}\left(
k;n,r\right)  $ includes the scaling effect $\left(  \exp_{q}\left(
C_{q}\right)  \right)  ^{1-q}=1+\left(  1-q\right)  C_{q}\left(  >0\right)  $
in itself. Of course, such a scaling effect disappears when $q\rightarrow1$.

\section{$\alpha$-divergence and large deviation estimate derived from the
generalized binomial distribution}

The $q$-divergence is explicitly derived from the definition
(\ref{def_lnq_q-binomial distribution}) of the $q$-binomial distribution
$b_{q}\left(  k;n,r\right)  $.

\begin{theorem}
For the $q$-binomial distribution $b_{q}\left(  k;n,r\right)  $ defined by
(\ref{def_lnq_q-binomial distribution}), we have%
\begin{equation}
\ln_{q}b_{q}\left(  k;n,r\right)  \simeq-\frac{n^{2-q}}{2-q}D_{2-q}\left(
p\left\Vert r\right.  \right)  +C_{q}\label{qGBin-qdiv}%
\end{equation}
for large $n\in\mathbb{N}$ where $D_{q}\left(  p\left\Vert r\right.  \right)
$ is the $q$-divergence defined by%
\begin{equation}
D_{q}\left(  p\left\Vert r\right.  \right)  :=\sum_{i=0}^{1}p_{i}\ln
_{2-q}\frac{p_{i}}{r_{i}}=\frac{1-\sum\limits_{i=0}^{1}p_{i}^{q}r_{i}^{1-q}%
}{1-q}%
\end{equation}
and%
\begin{align}
p &  :=\left(  p_{0},p_{1}\right)  =\left(  \frac{k}{n},1-\frac{k}{n}\right)
,\\
r &  :=\left(  r_{0},r_{1}\right)  =\left(  r,1-r\right)  .
\end{align}

\end{theorem}

\begin{proof}
A straightforward computation using (\ref{2-q-correspondence}) yields the right side of
(\ref{qGBin-qdiv}).
\end{proof}

Of course, for the $q$-multinomial distribution $m_{q}\left(  n_{1}%
,\cdots,n_{k};n,r_{1},\cdots,r_{k}\right)  $ defined by
(\ref{def_lnq_q-multinomial distribution}), we easily obtain the similar
result,%
\begin{equation}
\ln_{q}m_{q}\left(  n_{1},\cdots,n_{k};n,r_{1},\cdots,r_{k}\right)
\simeq-\frac{n^{2-q}}{2-q}D_{2-q}\left(  p\left\Vert r\right.  \right)
+C_{q}\label{qMulti-qdiv}%
\end{equation}
for large $n\in\mathbb{N}$ where $D_{q}\left(  p\left\Vert r\right.  \right)
$ is the $q$-divergence defined by%
\begin{equation}
D_{q}\left(  p\left\Vert r\right.  \right)  :=\sum_{i=1}^{k}p_{i}\ln
_{2-q}\frac{p_{i}}{r_{i}}=\frac{1-\sum\limits_{i=1}^{k}p_{i}^{q}r_{i}^{1-q}%
}{1-q}%
\end{equation}
and $p:=\left(  p_{1},\cdots,p_{k}\right)  =\left(  \frac{n_{1}}{n}%
,\cdots,\frac{n_{k}}{n}\right)  $ and $r:=\left(  r_{1},\cdots,r_{k}\right)  $.

(\ref{qMulti-qdiv}) recovers the standard case in the limit $q\rightarrow1$:%
\begin{equation}
\ln m_{1}\left(  n_{1},\cdots,n_{k};n,r_{1},\cdots,r_{k}\right)  \simeq
-nD_{1}\left(  p\left\Vert r\right.  \right)
\end{equation}
where $D_{1}\left(  p\left\Vert r\right.  \right)  $ is Kullback-Leibler (KL)
divergence defined by%
\begin{equation}
D_{1}\left(  p\left\Vert r\right.  \right)  :=\sum_{i=1}^{k}p_{i}\ln
\frac{p_{i}}{r_{i}}.
\end{equation}

The $q$-divergence is known to have the simple relation with the $\alpha
$-divergence \cite{Oha07}.

\begin{proposition}
The $\alpha$-divergence $D^{\left(  \alpha\right)  }\left(  p\left\Vert
r\right.  \right)  $ defined by%
\begin{equation}
D^{\left(  \alpha\right)  }\left(  p\left\Vert r\right.  \right)  :=\left\{
\begin{array}
[c]{ll}%
\frac{4}{1-\alpha^{2}}\left(  1-\sum\limits_{i}p_{i}^{\frac{1-\alpha}{2}}%
r_{i}^{\frac{1+\alpha}{2}}\right)  & \quad\left(  \alpha\neq\pm1\right) \\
\sum\limits_{i}r_{i}\ln\frac{r_{i}}{p_{i}} & \quad\left(  \alpha=1\right) \\
\sum\limits_{i}p_{i}\ln\frac{p_{i}}{r_{i}} & \quad\left(  \alpha=-1\right)
\end{array}
\right.  \label{alpha-divergence}%
\end{equation}
has the following simple relation with the $q$-divergence.
\begin{equation}
D^{\left(  \alpha\right)  }\left(  p\left\Vert r\right.  \right)  =\frac{1}%
{q}D_{q}\left(  p\left\Vert r\right.  \right)  \quad\left(  q\neq0,1\right)
\label{alpha-q-divergence-relation}%
\end{equation}
where%
\begin{equation}
q=\frac{1-\alpha}{2}\quad\left(  \alpha\neq\pm1\right)  .
\label{q-alpha-parameter}%
\end{equation}

\end{proposition}

The simple transformation (\ref{q-alpha-parameter}) of the parameters $q$ and
$\alpha$ reveals that the $q$-divergence is mathematically equivalent to the
$\alpha$-divergence. This means that the $\alpha$-divergence is derived from
the $q$-binomial distribution.

\begin{corollary}
For the $q$-binomial distribution $b_{q}\left(  k;n,r\right)  $ defined by
(\ref{def_lnq_q-binomial distribution}), we have%
\begin{equation}
\ln_{q}b_{q}\left(  k;n,r\right)  \simeq-n^{\frac{3+\alpha}{2}}D^{\left(
-2-\alpha\right)  }\left(  p\left\Vert r\right.  \right)  +C_{\alpha}^{\prime}%
\end{equation}
where $D^{\left(  \alpha\right)  }\left(  p\left\Vert r\right.  \right)  $ is
the $\alpha$-divergence defined by (\ref{alpha-divergence}) and $\lim
_{\alpha\rightarrow-1}C_{\alpha}^{\prime}=0$.
\end{corollary}

The $q$-divergence was introduced as the relative entropy in Tsallis
statistics \cite{Ts98}\cite{BPT98}, independently from the $\alpha$-divergence
and the family of the $f$-divergence. In fact, no references and comments on
these divergences were given in \cite{Ts98}\cite{BPT98}. On the other hand,
the $\alpha$-divergence has much longer history than the $q$-divergence
and was originally introduced in the evaluation of the classification errors
\cite{Cher52}. Later, the $\alpha$-divergence has been studied in the
information geometry for providing the geometrical structures of the manifold
of probability measures which is well consistent with the fundamentals in
statistics \cite{AN00}.

Finally, using (\ref{qGBin-qdiv}), the large deviation estimate will be
presented by means of the $q$-divergence (or the $\alpha$-divergence) in a
combinatorial way.

\begin{theorem}
Let $X_{i}$ $\left(  i=1,\cdots,n\right)  $ be a random variable taking values
in $\left\{  0,1\right\}  $ with probability
\begin{equation}
P\left(  X_{i}=0\right)  =r,\quad P\left(  X_{i}=1\right)  =1-r.
\end{equation}
If a sum of the random variables $\sum_{i=1}^{n}X_{i}$ follows the the
$q$-binomial distribution $b_{q}\left(  k;n,r\right)  $, for $0<x<r$ and
$0<q<2$ we have%
\begin{equation}
\frac{1}{n^{2-q}}\ln_{q}P\left(  \frac{1}{n}\sum_{i=1}^{n}X_{i}<x\right)
\simeq-\frac{1}{2-q}D_{2-q}\left(  x\left\Vert r\right.  \right)  .
\label{main result}%
\end{equation}

\end{theorem}

\begin{proof}%
\begin{equation}
P\left(  \frac{1}{n}\sum_{i=1}^{n}X_{i}<x\right)  =P\left(  \sum_{i=1}%
^{n}X_{i}<nx\right)  =\sum\limits_{k=0}^{\left\lfloor nx\right\rfloor }%
b_{q}\left(  k;n,r\right)
\end{equation}
where $\left\lfloor a\right\rfloor :=\max\left\{  m\in\mathbb{Z}\left\vert
m\leq a\right.  \right\}  $.

First, consider the upper bound. Each term in the sum is bounded by $\left.
b_{q}\left(  k;n,r\right)  \right\vert _{k=\left\lfloor nx\right\rfloor
}=b_{q}\left(  \left\lfloor nx\right\rfloor ;n,r\right)  $ and so%
\begin{equation}
P\left(  \frac{1}{n}\sum_{i=1}^{n}X_{i}<x\right)  \leq\left(  \left\lfloor
nx\right\rfloor +1\right)  b_{q}\left(  \left\lfloor nx\right\rfloor
;n,r\right)
\end{equation}
Let $A_{n}$ be defined by the right side of this inequality%
\begin{equation}
A_{n}:=\left(  \left\lfloor nx\right\rfloor +1\right)  b_{q}\left(
\left\lfloor nx\right\rfloor ;n,r\right)  .
\end{equation}
(i) If $q=1$, that is, the random variables $X_{i}$ are i.i.d., $\frac{1}%
{n}\ln A_{n}$ can be computed as%
\begin{align}
\frac{1}{n}\ln A_{n}  &  =\frac{1}{n}\ln\left(  \left\lfloor nx\right\rfloor
+1\right)  +\frac{1}{n}\ln b_{1}\left(  \left\lfloor nx\right\rfloor
;n,r\right) \\
&  \simeq-D_{1}\left(  x\left\Vert r\right.  \right)  \quad\left(
n\gg0\right)  .
\end{align}
Thus, we obtain the standard upper bound:%
\begin{equation}
\frac{1}{n}\ln P\left(  \frac{1}{n}\sum_{i=1}^{n}X_{i}<x\right)  \leq
-D_{1}\left(  x\left\Vert r\right.  \right)  .
\end{equation}
(ii) If $0<q<1$, we have%
\begin{align}
&  \frac{1}{n^{2-q}}\ln_{q}A_{n}\\
&  =\frac{\ln_{q}b_{q}\left(  \left\lfloor nx\right\rfloor ;n,r\right)
+\left(  b_{q}\left(  \left\lfloor nx\right\rfloor ;n,r\right)  \right)
^{1-q}\ln_{q}\left(  \left\lfloor nx\right\rfloor +1\right)  }{n^{2-q}}\\
&  \simeq-\frac{1}{2-q}D_{2-q}\left(  x\left\Vert r\right.  \right)
\quad\left(  n\gg0\right) \\
&  \quad\left(  \because\frac{1}{n^{2-q}}\ln_{q}\left(  \left\lfloor
nx\right\rfloor +1\right)  \simeq0\right)
\end{align}
Thus, we obtain the upper bound:%
\begin{equation}
\frac{1}{n^{2-q}}\ln_{q}P\left(  \frac{1}{n}\sum_{i=1}^{n}X_{i}<x\right)
\leq-\frac{1}{2-q}D_{2-q}\left(  x\left\Vert r\right.  \right)  .
\end{equation}
(iii) If $1<q<2$, using $\ln_{q}a<\ln a$, we have%
\begin{align}
&  \frac{1}{n^{2-q}}\ln_{q}A_{n}<\frac{1}{n^{2-q}}\ln A_{n}\\
&  =\frac{1}{n^{2-q}}\left(  \ln\left(  \left\lfloor nx\right\rfloor
+1\right)  +\ln b_{q}\left(  \left\lfloor nx\right\rfloor ;n,r\right)  \right)
\\
&  \simeq\frac{1}{n^{2-q}}\ln\exp_{q}\left(  -\frac{n^{2-q}}{2-q}%
D_{2-q}\left(  x\left\Vert r\right.  \right)  +C_{q}\right)  \quad\left(
n\gg0\right) \\
&  =\frac{1}{n^{2-q}}\frac{\ln\left[  1+\left(  q-1\right)  \left(
\frac{n^{2-q}}{2-q}D_{2-q}\left(  x\left\Vert r\right.  \right)
+C_{q}\right)  \right]  }{1-q}\\
&  <-\frac{1}{2-q}D_{2-q}\left(  x\left\Vert r\right.  \right)  \quad\left(
\because\ln a<a-1\right)
\end{align}
Thus, we obtain the upper bound:%
\begin{equation}
\frac{1}{n^{2-q}}\ln_{q}P\left(  \frac{1}{n}\sum_{i=1}^{n}X_{i}<x\right)
\leq-\frac{1}{2-q}D_{2-q}\left(  x\left\Vert r\right.  \right)  .
\label{q-upper}%
\end{equation}

Next, consider the lower bound. If $x<r$, obviously%
\begin{equation}
P\left(  \frac{1}{n}\sum_{i=1}^{n}X_{i}<x\right)  \geq b_{q}\left(
\left\lfloor nx\right\rfloor ;n,r\right)  .
\end{equation}
Then, immediately we find%
\begin{align}
\frac{1}{n^{2-q}}\ln_{q}P\left(  \frac{1}{n}\sum_{i=1}^{n}X_{i}<x\right)   &
\geq\frac{1}{n^{2-q}}\ln_{q}b_{q}\left(  \left\lfloor nx\right\rfloor
;n,r\right)  \nonumber\\
&  \simeq-\frac{1}{2-q}D_{2-q}\left(  x\left\Vert r\right.  \right)
.\label{q-lower}%
\end{align}

The upper bound (\ref{q-upper}) and lower bound (\ref{q-lower}) lead to the
main result (\ref{main result}).
\end{proof}

\section{Conclusion}

Starting from the fundamental nonlinear differential equation $dy/dx=y^{q}$
only, we obtain the large deviation estimate in Tsallis statistics by means of
the $q$-divergence (or the $\alpha$-divergence) as the generalized rate
function.
The present approach provides us with important formulas for a power-law system
such as $q$-logarithm, $q$-exponential,
$q$-product, $q$-Gaussian distribution, $q$-Stirling's formula, $q$%
-multinomial coefficient, Tsallis entropy, and $q$-binomial distribution. We
use analytical derivations only, which recover the standard case when
$q\rightarrow1$. Therefore, our present results reveal that there exists a
fundamental and novel mathematical structure for power-law system recovering
the standard case (i.e., i.i.d. case in probability theory, Shannon information theory,
Boltzmann-Gibbs statistics in statistical physics) as a
special case. Of course, some important theorems in a power-law system are
still missing. However, the present strategy will definitely provide us
with fruitful applications in any related areas. Such
some applications will be presented in the near future.

\section*{Acknowledgment}
The first author (H.S.) is very grateful to Prof. Jan Naudts for the fruitful
discussions about the LDP during his sabbatical stay in Antwerp Univ.




\begin{thebibliography}{99}

\bibitem {Varadhan66}S.R.S. Varadhan, Asymptotic probabilities and
differential equations, Commun. Pure Appl. Math., 19, 261-286 (1966).

\bibitem {Varadhan84}S.R.S. Varadhan, Large deviations and applications, SIAM,
Philadelphia (1984).

\bibitem {DZ98}A. Dembo and O. Zeitouni, Large deviations techniques and
applications, 2nd ed., Springer (1998).

\bibitem {Ellis85}R.S. Ellis, Entropy, Large deviations, and statistical
mechanics, Springer (1985).

\bibitem {Touchette09}H. Touchette, The large deviation approach to
statistical mechanics, Phys. Rep. 478, 1-69 (2009).

\bibitem {AK96}S. Asmussen, C. Kl\"{u}pelberg, Large deviations results for
subexponential tails, with applications to insurance risk, Stoch. Proc. Appl.
64, 103-125, (1996).

\bibitem {CoverThomas06}T.M. Cover and J.A. Thomas, Elements of Information
Theory, 2nd ed., New York, Wiley (2006).

\bibitem {Suyari06}H. Suyari, Mathematical structures derived from the
$q$-multinomial coefficient in Tsallis statistics, Physica A, 368, pp.63-82 (2006).

\bibitem {RuizTsallis2012}G. Ruiz and C. Tsallis, Towards a large deviation
theory for strongly correlated systems, Phys. Lett. A 376, 2451-2454 (2012).

\bibitem {NaudtsSuyari2014}J. Naudts and H. Suyari, Large deviation estimates
involving deformed exponential functions, arXiv:1404.3410.

\bibitem {SuyariWada08}H. Suyari and T. Wada, Multiplicative duality,
$q$-triplet and $\left(  \mu,\nu,q\right)  $-relation derived from the
one-to-one correspondence between the $\left(  \mu,\nu\right)  $-multinomial
coefficient and Tsallis entropy $S_{q}$, Physica A, vol.387, 71-83 (2008).

\bibitem {Suyari04-LawofError}H. Suyari and M. Tsukada, Law of error\ in
Tsallis statistics, IEEE Trans. Inform. Theory, 51, 753-757 (2005).

\bibitem {WS05}T. Wada and A.M. Scarfone, Connections between Tsallis'
formalisms employing the standard linear average energy and ones employing the
normalized $q$-average energy, Phys. Lett. A, vol.335, 351-362 (2005).


\bibitem {Ts88}C. Tsallis, Possible generalization of Boltzmann-Gibbs
statistics, J. Stat. Phys. 52, 479-487 (1988).


\bibitem {Tsallis09}C. Tsallis, Introduction to Nonextensive Statistical
Mechanics: Approaching a Complex World, Springer, New York (2009).

\bibitem {Naudts11}J. Naudts, Generalised Thermostatistics, Springer, London (2011).

\bibitem {NMW03}L. Nivanen, A. Le Mehaute, Q.A. Wang, Generalized algebra
within a nonextensive statistics, Rep. Math. Phys. 52, 437-444 (2003).

\bibitem {Bo03}E.P. Borges, A possible deformed algebra and calculus inspired
in nonextensive thermostatistics, Physica A 340, 95-101 (2004).

\bibitem {Scarfone13}A.M. Scarfone, Entropic forms and related algebras, Entropy, 15, 624-649 (2013).

\bibitem {Ts98}C. Tsallis, Generalized entropy-based criterion for consistent
testing, Phys.Rev.E, vol.58, 1442-1445 (1998).

\bibitem {BPT98}L. Borland, A.R. Plastino and C. Tsallis, Information gain
within nonextensive thermostatistics, J.Math.Phys. vol.39, 6490-6501 (1998);
[Errata: J.Math.Phys. 40, 2196 (1999)].

\bibitem {Oha07}A. Ohara, Geometry of distributions associated with Tsallis
statistics and properties of relative entropy minimization, Phys.Lett. A,
vol.370, 184-193 (2007).

\bibitem {Cher52}H. Chernoff, A measure of asymptotic efficiency for tests of
a hypothesis based on the sum of observations, Ann. Math. Statist., 23,
493-655 (1952).

\bibitem {AN00}S. Amari and H. Nagaoka, Methods of information geometry,
Trans. Math. Monogr., vol. 191, Amer. Math. Soc. \& Oxford Univ. Press (2000).


\end{thebibliography}
%

\end{document}